\documentclass[conference,letterpaper]{IEEEtran}
\usepackage[utf8]{inputenc} 
\usepackage[T1]{fontenc}
\usepackage{url}
\usepackage{ifthen}
\usepackage{cite}
\usepackage[cmex10]{amsmath} 
\usepackage{amssymb}  
\usepackage{amsthm}
\usepackage{enumerate,color,float}
\usepackage[justification=centering]{caption}
\usepackage{amsfonts}
\usepackage{mathtools}
\usepackage{arydshln}
\usepackage{stmaryrd}
\usepackage{cases}
\usepackage{multicol}
\usepackage{algorithm,algorithmic}
\usepackage{tikz,pgfplots,rotating}
\usepackage{bm}
\newtheoremstyle{remboldstyle}
  {}{}{}{}{\bfseries}{.}{.5em}{{\thmname{#1 }}{\thmnumber{#2}}}
\theoremstyle{plain}
\newtheorem{theorem}{Theorem}

\theoremstyle{definition}
\newtheorem{definition}[theorem]{Definition}

\newtheorem{example}[theorem]{Example}

\theoremstyle{remboldstyle}

\theoremstyle{corollary}

\theoremstyle{lemma}

\newcommand{\mb}[1]{\mathbf{#1}}
\newcommand{\mbb}[1]{\mathbb{#1}}

\newcommand{\boldV}[1]{\boldsymbol{\mathbf{#1}}}

\newcommand{\Rm}[1]{\mathrm{#1}}

\newcommand{\defeq}{\triangleq}

\DeclareMathOperator*{\argmin}{\arg\,\min}
\DeclareMathOperator*{\argmax}{\arg\,\max}

\DeclareMathOperator{\Wt}{wt}

\newcommand{\GF}[1]{\mbb{F}_{#1}}

\newcommand{\om}{\omega}

\newcommand{\ve}{\mathbf{e}}

\newcommand{\vt}{\mathbf{t}}
\newcommand{\vs}{\mathbf{s}}
\newcommand{\vv}{\mathbf{v}}

\newcommand{\cB}{\mathcal{B}}

\newcommand{\cC}{\mathcal{C}}
\newcommand{\cD}{\mathcal{D}}
\newcommand{\cE}{\mathcal{E}}

\newcommand{\cJ}{\mathcal{J}}
\newcommand{\cN}{\mathcal{N}}

\newcommand{\tr}{\mathsf{T}}

\newcommand{\qedblack}{\hspace*{\fill}
  $\blacksquare$\smallskip}

\newcommand{\defend}{\qedblack}

\newcommand{\Footnotetext}[2]{\begin{figure}[!b]\footnotesize%
  \vspace{-3ex}\hrulefill\hfill\makebox[0em]{}\hfill\makebox[0em]{}%
  \par${}^{#1}$ #2\vspace{-0.60ex}\end{figure}\addtocounter{figure}{0}}

\begin{document}
\title{Pseudocodeword-based Decoding \\of Quantum Stabilizer Codes}

\author{

  \IEEEauthorblockN{July X.\ Li and 
   						Pascal O.\ Vontobel}
  \IEEEauthorblockA{Department of Information Engineering,
                    The Chinese University of Hong Kong, 
                    Hong Kong\\
                    Email: \{july.x.li, 
                    		  pascal.vontobel\}@ieee.org\vspace{-0.5cm}
}
}

\maketitle


\begin{abstract}
  It has been shown that graph-cover pseudocodewords can be used to
  characterize the behavior of sum-product algorithm (SPA) decoding of
  classical codes. In this paper, we leverage and adapt these results to
  analyze SPA decoding of quantum stabilizer codes. We use the obtained
  insights to formulate modifications to the SPA that overcome some of its
  weaknesses.
\end{abstract}


\Footnotetext{}{Supported in part by RGC GRF grant 2150965.}


\section{Introduction}
\label{sec:introduction:1}

Graph covers have been shown to be a useful tool for analyzing sum-product
algorithm (SPA) decoding of classical codes~\cite{vontobel2013counting}. The
task of analyzing the behavior of SPA decoding for quantum stabilizer codes is
more challenging, especially because the degeneracy of quantum stabilizer
codes needs to be taken into account. Despite these challenges, being able to
understand and improve the behavior of the SPA is highly desirable, since it
has been observed that the performance of the SPA is far from satisfactory
when decoding quantum stabilizer codes of high degeneracy (see, e.g., the
discussion of simulation results of various massage-passing iterative decoding
algorithms and LP decoders in~\cite{babar2015fifteen, li2018lp,
  liu2018neural}).

In this paper, in a first step, we use graph-cover pseudocodewords to analyze
the behavior of SPA. In particular, we can show that the decoding ability of
the SPA is limited by the minimum distance of the normalizer label code, which
is a serious problem for quantum LDPC codes, e.g., the toric
codes~\cite{tillich2014quantum} and MacKay's bicycle
codes~\cite{mackay2004sparse}, where the minimum distance of the normalizer
label code is no larger than the row weight of its parity-check matrix due to
the self-orthogonality of the stabilizer label code.

In a second step, we use the obtained insights to formulate modifications to
the SPA that overcome some of its weaknesses. Taking advantage of the
degeneracy of the quantum stabilizer code, the performance of the decoder is
then limited by the minimum distance of the quantum stabilizer code $d$ instead of
the minimum distance of the normalizer label code $d_{\cN}$. For notational details, see Section~\ref{sec:notation:2}.

This paper is organized as follows. In Section~\ref{sec:notation:1}, we review
some basic notations including the stabilizer formalism and the standard SPA
for quantum stabilizer codes. In Section~\ref{sec:theoretical:analysis:1}, we
analyze the performance of SPA for quantum stabilizer codes and give some
other theoretical results about degenerate decoders of quantum stabilizer
codes. In Section~\ref{sec:pcw:decoding:0}, we propose, first, some methods to
improve the performance of the SPA for general quantum stabilizer codes and,
second, a pseudocodeword-based decoder for quantum cycle codes. Finally, we
show some simulation results in Section~\ref{sec:simulation:results:1}.


\section{Basics}
\label{sec:notation:1}

\subsection{Quantum Stabilizer Formalism}
\label{sec:notation:2}
We refer the readers to~\cite{nielsen2010quantum,lidar2013quantum} for a
detailed introduction to quantum stabilizer codes, some recent developments of
quantum error-correction codes, and more details of the notations. Moreover,
see~\cite{li2018lp} for the use of pseudocodewords in the context of quantum
stabilizer codes. Due to the page limitations, we only introduce the essential
notations which are used throughout the paper.

Consider an $\llbracket n,k,d\rrbracket$ quantum stabilizer code $\cC$ of
length $n$, dimension $k$, and minimum distance $d$.  The quantum stabilizer
code $\cC$ may be characterized using the equivalent binary representation of
its stabilizer, namely its binary stabilizer label code $\cB$, which is
self-orthogonal under the symplectic inner product to guarantee the
commutativity of the generators of the stabilizer. The binary representation of a
Pauli operator on $n$ qubits is a length-$2n$ binary vector $\vv = [\vv_1,
..., \vv_n] \in \left(\GF{2}^2\right)^n$, where each $\vv_i$ is obtained by
mapping $I$, $X$, $Y$, and $Z$ onto $\GF{2}^2$ as follows
\begin{equation*}
  I \mapsto [0,0], \ 
  X \mapsto [1,0], \ 
  Y \mapsto [1,1], \text{ and }
  Z \mapsto [0,1],
\end{equation*}
and the weights of them are defined to be, respectively, 
\begin{equation*}
  \Wt([0,0]) \defeq 0 \text{ and }
  \Wt([1,0]) = \Wt([1,1]) = \Wt([0,1]) \defeq 1.
\end{equation*}
In this paper, we make the following assumptions:
\begin{itemize}

\item the normalizer label code $\cN$ is the dual code of $\cB$ under the
  symplectic inner product (note that the self-orthogonality of $\cB$ implies
  that $\cB \subseteq \cN$);

\item both $\cB$ and $\cN$ are binary linear codes of length $2n$ and of
  dimension $n-k$ and $n+k$, respectively;

\item the weight of $\vv$ is $\Wt(\vv) \defeq \sum_{i}\Wt(\vv_i)$;

\item $d \defeq \min_{\vv\in \cN \backslash \cB} \Wt(\vv)$ and $t \defeq
  \left\lfloor \frac{d-1}{2} \right\rfloor$;

\item $d_{\cN} \defeq \min_{\vv\in \cN} \Wt(\vv)$ and $t_{\cN} \defeq
  \left\lfloor \frac{d_{\cN}-1}{2} \right\rfloor$.

\end{itemize}

A quantum stabilizer code $\cC$ is called a quantum cycle code if its
normalizer label code $\cN$ is a cycle code, which means that the number of
$1$'s per column of the parity-check matrix $H$ describing $\cN$ is two.  For
example, the toric codes are quantum cycle
codes (see, e.g.,~\cite{tillich2014quantum,Li:Vontobel:16:1}).

The quantum channel that we use in this paper is the quantum depolarizing
channel (QDCh). Similar to the binary symmetric channel (BSC), the action of a
QDCh with depolarizing probability $p$ is such that it acts independently on
each qubit: a qubit is either unchanged with probability $1-p$, or affected by
a unitary operator $X$, $Y$, or $Z$, each with probability $p/3$. Since we are
decoding with respect to \emph{binary} normalizer label codes, decoding is
based on approximating the QDCh by two independent BSCs with crossover
probability $2p/3$, i.e., the probability for having a bit-flip and a
phase-flip is $2p/3$ independently for each qubit.


\begin{definition}
\label{def:ND_DD}

Given a syndrome $\vs\in\GF{2}^{n-k}$, 
let $\vs \mapsto \vt(\vs)$ be the mapping giving a coset representative of the coset of $\cN$ corresponding to the syndrome $\vs$.
Note that if $\ve$ is the binary representation of the actual error, then $H \ve^{\tr} = \vs^{\tr}$
and $\ve \in \vt(\vs)+\cN$.

A non-degenerate decoder $\cD_{\Rm{ND}}$ outputs a vector based on the syndrome $\vs$;  
an error vector $\vv$ leads to a decoding error for $\cD_{\Rm{ND}}$ if $\vv \neq \cD_{\Rm{ND}}(\vv H^{\tr})$.
A degenerate decoder $\cD_{\Rm{D}}$ outputs a coset of $\cB$ based on the syndrome $\vs$; 
an error vector $\vv$ leads to a decoding error for $\cD_{\Rm{D}}$ if $\vv\notin \cD_{\Rm{D}}(\vv H^{\tr})$.
The blockwise ML (non-)degenerate decoders 
$\cD_{\Rm{ND}}^{\Rm{ML}}$, $\cD_{\Rm{D}}^{\Rm{ML}*}$, and $\cD_{\Rm{D}}^{\Rm{ML}}$ are defined to be, respectively,
\begin{align*}
&\cD_{\Rm{ND}}^{\Rm{ML}}(\vs) \defeq \argmin_{\vv \in\mb{t}(\vs) + \cN} \Wt(\vv),\ \ \ 
\cD_{\Rm{D}}^{\Rm{ML}*}(\vs) \defeq \cD_{\Rm{ND}}^{\Rm{ML}}(\vs) +\cB,\\
&\cD_{\Rm{D}}^{\Rm{ML}}(\vs) \defeq \argmax_{\boldV{\ell}+\cB:\ \boldV{\ell} \in\mb{t}(\vs) + \cN} p(\boldV{\ell}+\cB|\vs),
\end{align*}
where $p(\boldV{\ell}+\cB|\vs)$ is the probability of the coset $\boldV{\ell}+\cB$ based on the syndrome $\vs$.
\defend
\end{definition}

For the simulations in this paper, there is a decoding error if the output vector is not in the same coset of $\cB$ as the actual error or the output coset is not the same coset of $\cB$ as the coset of the actual error.

\subsection{SPA decoding, graph covers, and pseudocodewords}

SPA decoding of a quantum stabilizer code $\cC$ consists of the following
steps: 1) running the SPA on a factor graph representing a coset of the
normalizer label code $\cN$, where the coset is defined by the syndrome $\vs$
that is obtained from suitable quantum measurements; 2) outputting a vector
$\vv$, 3) finding the coset of $\cB$ containing $\vv$. (For further details,
see, e.g., \cite[Section IV]{babar2015fifteen}.) In this paper, the factor
graphs are normal factor graphs, where variables are associated with edges.

It was shown in~\cite{Yedidia:Freeman:Weiss:05:1} that fixed points of the SPA
correspond to stationary points of the Bethe free energy function. As
discussed in~\cite{vontobel2013counting}, for LDPC codes this means that the
beliefs obtained at a fixed point of the SPA induce a pseudocodeword $\boldV{\om}$. 
For example, if we consider a binary linear code, then the $i^{\Rm{th}}$ component of $\boldV{\om}$ is $\om_i \defeq b_i(1)$ assuming the belief of the $i^{\Rm{th}}$ variable is $[b_i(0), b_i(1)]$.
The paper~\cite{vontobel2013counting} also introduced the symbolwise graph-cover
decoder, a decoder that finds the pseudocodeword with minimal Bethe free
energy, or, equivalently, the pseudocodoword with the most pre-images in all
$M$-covers of the base normal factor graph (after properly discounting for a
channel-output-dependent term), when $M$ goes to infinity. For general codes,
symbolwise graph-cover decoding is an approximation of the true behavior of
SPA decoding. However, for cycle codes it was shown
in~\cite{pfister2013relevance} that SPA decoding is equivalent to symbolwise
graph-cover decoding. Note that, although symbolwise graph-cover decoding is
based on $M$-covers where $M$ goes to infinity, in many instances the study of 
pseudocodewords induced by codewords in $M$-covers for small $M$ gives already
many insights into the suboptimality of SPA decoding (see, e.g., the upcoming
Fig.~\ref{fig:PSPA_toric2_2cover} that shows an $M$-cover for $M = 2$).


\section{Theoretical Analysis}
\label{sec:theoretical:analysis:1}

In this section we characterize the performance of the non-degenerate and
degenerate decoders defined in Definition~\ref{def:ND_DD}. In particular, in
Theorems~\ref{thm:ND_decoder} and~\ref{thm:DD_decoder} we prove that the
minimum weight of errors that the non-degenerate and degenerate
decoders fail to decode are $t_{\cN}+1$ and $t+1$, respectively. Moreover, in
Theorems~\ref{thm:SPA_half_d_N} and~\ref{thm:SPA_d_N}, we show two types of
decoding errors limiting the performance of SPA decoding of quantum cycle
codes.


\begin{theorem}
  \label{thm:ND_decoder}
  The minimum weight of errors leading to decoding errors for $\cD_{\Rm{ND}}^{\Rm{ML}}$ is $t_{\cN}+1$.
\end{theorem}

\begin{proof}
  See Appendix~\ref{app:ND_decoder_proof}.  
\end{proof}


\begin{theorem}
  \label{thm:DD_decoder}
  The minimum weights of errors leading to decoding errors for $\cD_{\Rm{D}}^{\Rm{ML}}$ and $\cD_{\Rm{D}}^{\Rm{ML}*}$ both are $t+1$.
\end{theorem}

\begin{proof}
 See Appendix~\ref{app:DD_decoder_proof}.
\end{proof}

\begin{theorem}
  \label{thm:SPA_half_d_N}
  The minimum weight of errors that the SPA fails to decode for a toric code
  is $2$. For a $\llbracket 2L^2,2,L\rrbracket$ toric code with $L\geq 5$, 
  the number of such weight-$2$ errors is $12L^2$. 
\end{theorem}

\begin{proof}
  For a $\llbracket 2L^2,2,L\rrbracket$ toric code $\cC$ with $L < 5$, it cannot 
  correct some weight-2 errors because of its minimum distance.
  For a $\llbracket 2L^2,2,L\rrbracket$ toric code $\cC$ with $L \geq 5$, 
  there are two types of weight-$2$ errors that cannot be corrected using SPA decoding as shown in
  Fig.~\ref{fig:PSPA_toric2}. (Here and for other similar figures we use the
  drawing conventions listed in Table~\ref{table:drawing:conventions:1};
  moreover, edges with components close to $0$ are not drawn). 
  For both cases, we obtain an SPA pseudocodeword $\boldV{\om} = [\om,\om,\om,\om]$, 
  for some $\om\in(0,1]$. 
  When the SPA decoder makes hard decisions based on $\boldV{\om}$, it outputs either 
  $[0,0,0,0]$ or $[1,1,1,1]$ and hence fails to match the syndrome.
  The minimum weight of errors resulting in decoding failures for toric codes
  is $2$, since weight-$1$ errors can be corrected. 
  If we count the number of such weight-$2$ errors, there
  are $6$ in each length-$4$ cycle and $12L^2$ in total.
\end{proof}

For a quantum cycle code with even $d_{\cN}$, the minimum weight of
errors that the SPA fails to decode is no larger than $d_{\cN}/2$ because of similar
problems as in Fig.~\ref{fig:PSPA_toric2}.


\begin{theorem}
  \label{thm:SPA_d_N}
  The minimum weight of errors that the SPA fails to decode for a toric code is no
  larger than $d_{\cN}$.
\end{theorem}

\begin{proof}
  We want to show that there exist errors of weight $d_{\cN}$ that the SPA fails
  to decode.  Since the minimum weight of vectors in the normalizer label code
  $\cN$ is $d_{\cN}$, there exists a cycle of length $d_{\cN}$ in $\cN$ and we
  assume that the error is a path of length $d_{\cN}$ starting from any check
  involved in that cycle. Fig.~\ref{fig:PSPA_toric2_2cover} is a $2$-cover of
  the relevant part of a toric code. We claim that the SPA cannot decode the above-mentioned error. 
  The reason is as follows. 
 There are two valid configurations in the 2-cover, where the red one can be projected down as a codeword with a valid syndrome, while the blue one cannot.
  The components of the pseudocodewords resulting from these valid configurations are shown next to the
  corresponding edges in Fig.~\ref{fig:PSPA_toric2_2cover}.
  The SPA pseudocodeword is a linear combination of such pseudocodewords, e.g., a rescaled SPA pseudocodeword in Fig.~\ref{fig:PSPA_toric9}, and the SPA decoder fails to output a vector with a valid syndrome no matter how to scale such SPA pseudocodeword.
\end{proof}

More generally, for quantum cycle codes, the SPA fails to decode errors of minimum weight no larger than $d_{\cN}$ for similar reasons.

\begin{table}[t]
  \caption{Drawing conventions for figures.}
  \label{table:drawing:conventions:1}
  \begin{center}
    \begin{tabular}{|l|l|}
      \hline
      empty vertex   & $s_i = 0$ for syndrome bit associated with $i$-th
      parity check \\
      \hline
      filled vertex  & $s_i = 1$ for syndrome bit associated with $i$-th
      parity check \\
      \hline
      black edge & channel introduced no error at that location \\
      \hline
      red edge & channel introduced an error at that location \\
      \hline
    \end{tabular}
  \end{center}
\end{table}

\begin{figure}[t]
  \begin{center}
    \begin{tikzpicture}[x=1cm, y= 1cm, node/.style={draw=none, inner sep = 0pt},
                          labeledh/.style={midway,above=0cm}, 
                          labeledhb/.style={midway,below=0cm},
                          labeledvl/.style={midway,left=0cm},
                          labeledvr/.style={midway,right=0cm},
                          check/.style={draw,fill=none,minimum size = 6pt,inner sep = 0pt},
                          checks/.style={draw,fill=black,minimum size = 6pt,inner sep = 0pt}
                          ]
	
\def\x{0pt}
\def\y{0pt}
\def\xa{\x+40pt}
\def\xb{\xa+40pt}
\def\xc{\xb+40pt}
\def\xd{\xc+40pt}
\def\xe{\xd+40pt}
\def\xf{\xe+40pt}
\def\ya{\y-40pt}
\def\yb{\ya-40pt}
\def\yc{\yb-40pt}
\def\yd{\yc-40pt}
\def\ye{\yd-40pt}
\def\yf{\ye-40pt}

\node[checks] (c11) at (\xa,\ya) {};
\node[check] (c12) at (\xb,\ya) {};

\node[check] (c21) at (\xa,\yb) {};
\node[checks] (c22) at (\xb,\yb) {};

\path[draw, thick, red] (c11) -- (c12) node[labeledhb] {$1$};
\path[draw, thick] (c21) -- (c22) node[labeledh] {$1$};

\path[draw, thick] (c11) -- (c21) node[labeledvl] {$1$};
\path[draw, thick, red] (c12) -- (c22) node[labeledvr] {$1$};

\node[checks] (c13) at (\xc,\ya) {};
\node[checks] (c14) at (\xd,\ya) {};

\node[checks] (c23) at (\xc,\yb) {};
\node[checks] (c24) at (\xd,\yb) {};

\path[draw, thick, red] (c13) -- (c14) node[labeledhb] {$1$};
\path[draw, thick, red] (c23) -- (c24) node[labeledh] {$1$};

\path[draw, thick] (c13) -- (c23) node[labeledvl] {$1$};
\path[draw, thick] (c14) -- (c24) node[labeledvr] {$1$};
\end{tikzpicture}
  \end{center}
  \caption{Rescaled pseudocodewords of a toric code.}
  \label{fig:PSPA_toric2}

  \begin{center}
    \begin{tikzpicture}[x=1cm, y= 1cm, node/.style={draw=none, inner sep = 0pt},
                          labeledha/.style={midway,above=0cm}, 
                          labeledhb/.style={midway,below=0cm},
                          labeledvl/.style={midway,left=0cm},
                          labeledvr/.style={midway,right=0cm},
                          check/.style={draw,fill=none,minimum size = 6pt,inner sep = 0pt},
                          checks/.style={draw,fill=black,minimum size = 6pt,inner sep = 0pt}
                          ]
	
\def\x{0pt}
\def\y{0pt}
\def\xa{\x+40pt}
\def\xb{\xa+40pt}
\def\xc{\xb+40pt}
\def\xd{\xc+40pt}
\def\xe{\xd+40pt}
\def\xf{\xe+40pt}
\def\ya{\y-40pt}
\def\yb{\ya-40pt}
\def\yc{\yb-40pt}
\def\yd{\yc-40pt}
\def\ye{\yd-40pt}
\def\yf{\ye-40pt}

\node[check] (c11) at (\xa,\ya) {};
\node[check] (cc11) at (\xa-6pt,\ya+6pt) {};

\node[check] (c12) at (\xb,\ya) {};
\node[check] (cc12) at (\xb-6pt,\ya+6pt) {};

\node[check] (c14) at (\xd,\ya) {};
\node[check] (cc14) at (\xd-6pt,\ya+6pt) {};

\node[check] (c15) at (\xe,\ya) {};
\node[check] (cc15) at (\xe-6pt,\ya+6pt) {};

\node[check] (c21) at (\xa,\yb) {};
\node[check] (cc21) at (\xa-6pt,\yb+6pt) {};

\node[checks] (c22) at (\xb,\yb) {};
\node[checks] (cc22) at (\xb-6pt,\yb+6pt) {};

\node (c23) at (\xc,\yb) {\color{red}...};
\node (cc23) at (\xc-6pt,\yb+6pt) {\color{red}...};

\node[checks] (c24) at (\xd,\yb) {};
\node[checks] (cc24) at (\xd-6pt,\yb+6pt) {};

\node[check] (c25) at (\xe,\yb) {};
\node[check] (cc25) at (\xe-6pt,\yb+6pt) {};

\path[draw, thick] (c11) -- (cc12) node {};
\path[draw, line width=0.4mm, blue] (cc11) -- (c12) node[labeledhb] {$\frac{1}{2}$};

\path[draw, thick] (c14) -- (cc15) node {};
\path[draw, line width=0.4mm, blue] (cc14) -- (c15) node[labeledhb] {$\frac{1}{2}$};

\path[draw, thick] (c11) -- (c21) node {};
\path[draw, line width=0.4mm, blue] (cc11) -- (cc21) node[labeledvl] {$\frac{1}{2}$};

\path[draw, thick] (cc12) -- (cc22) node {};
\path[draw, line width=0.4mm, blue] (c12) -- (c22) node[labeledvr] {\color{blue}$\frac{1}{2}$};

\path[draw, thick] (c14) -- (c24) node {};
\path[draw, line width=0.4mm, blue] (cc14) -- (cc24) node[labeledvl] {$\frac{1}{2}$};

\path[draw, thick] (cc15) -- (cc25) node {};
\path[draw, line width=0.4mm, blue] (c15) -- (c25) node[labeledvr] {$\frac{1}{2}$};

\path[draw, thick] (c21) -- (c22) node {};
\path[draw, line width=0.4mm, blue] (cc21) -- (cc22) node[labeledha] {$\frac{1}{2}$};

\path[draw, thick, red] (c22) -- (c23) node {};
\path[draw, thick, red] (cc22) -- (cc23) node[labeledha] {$1$};

\path[draw, thick, red] (c23) -- (c24) node {};
\path[draw, thick, red] (cc23) -- (cc24) node[labeledha] {$1$};

\path[draw, line width=0.4mm, blue] (c24) -- (c25) node {};
\path[draw, thick] (cc24) -- (cc25) node[labeledha] {\color{blue}$\frac{1}{2}$};



\end{tikzpicture}
  \end{center}
  \caption{Pseudocodewords (blue or red) of a toric code.}
  \label{fig:PSPA_toric2_2cover}

  \begin{center}
   \resizebox{1\columnwidth}{!}{\begin{tikzpicture}[x=1cm, y= 1cm, node/.style={draw=none, inner sep = 0pt},
                          labeledha/.style={midway,above=-2pt}, 
                          labeledhb/.style={midway,below=-2pt}, 
                          labeledvl/.style={midway,left=-2pt},
                          labeledvr/.style={midway,right=-2pt},
                          check/.style={draw,fill=none,minimum size = 6pt,inner sep = 0pt},
                          checks/.style={draw,fill=black,minimum size = 6pt,inner sep = 0pt}
                          ]
	
\def\x{0pt}
\def\y{0pt}
\def\xa{\x+40pt}
\def\xb{\xa+40pt}
\def\xc{\xb+40pt}
\def\xd{\xc+40pt}
\def\xe{\xd+40pt}
\def\xf{\xe+40pt}
\def\ya{\y-40pt}
\def\yb{\ya-40pt}
\def\yc{\yb-40pt}
\def\yd{\yc-40pt}
\def\ye{\yd-40pt}
\def\yf{\ye-40pt}

\node[check] (c11) at (\xa,\ya) {};
\node[check] (c12) at (\xb,\ya) {};
\node[check] (c13) at (\xc,\ya) {};

\node[check] (c21) at (\xa,\yb) {};
\node[checks] (c22) at (\xb,\yb) {};
\node[check] (c23) at (\xc,\yb) {};
\node[check] (c24) at (\xd,\yb) {};
\node[check] (c25) at (\xe,\yb) {};
\node[check] (c26) at (\xf,\yb) {};

\node[check] (c31) at (\xa,\yc) {};
\node[check] (c32) at (\xb,\yc) {};
\node[check] (c33) at (\xc,\yc) {};
\node[check] (c34) at (\xd,\yc) {};
\node[checks] (c35) at (\xe,\yc) {};
\node[check] (c36) at (\xf,\yc) {};

\node[check] (c44) at (\xd,\yd) {};
\node[check] (c45) at (\xe,\yd) {};
\node[check] (c46) at (\xf,\yd) {};

\path[draw, thick] (c11) -- (c12) node[labeledha] 		{\small $1$}; 
\path[draw, thick] (c12) -- (c13) node[labeledha] 		{\small $1$};
\path[draw, thick] (c11) -- (c12) node[labeledhb] 		{\scriptsize $(1)$};
\path[draw, thick] (c12) -- (c13) node[labeledhb] 		{\scriptsize $(2)$};

\path[draw, thick] (c11) -- (c21) node[labeledvr] 		{\small $1$};
\path[draw, thick] (c12) -- (c22) node[labeledvr] 		{\small $2$};
\path[draw, thick] (c13) -- (c23) node[labeledvr] 		{\small $1$};
\path[draw, thick] (c11) -- (c21) node[labeledvl] 		{\scriptsize $(3)$};
\path[draw, thick] (c12) -- (c22) node[labeledvl] 		{\scriptsize $(4)$};
\path[draw, thick] (c13) -- (c23) node[labeledvl] 		{\scriptsize $(5)$};

\path[draw, thick] (c21) -- (c22) node[labeledha] 		{\small $2$};
\path[draw, thick, red] (c22) -- (c23) node[labeledha] {\small $5$};
\path[draw, thick, red] (c23) -- (c24) node[labeledha] {\small $2$};
\path[draw, thick, red] (c24) -- (c25) node[labeledha] {\small $2$};
\path[draw, thick] (c25) -- (c26) node[labeledha] 		{\small $1$};
\path[draw, thick] (c21) -- (c22) node[labeledhb] 		{\scriptsize $(6)$};
\path[draw, thick, red] (c22) -- (c23) node[labeledhb] {\scriptsize $(7)$};
\path[draw, thick, red] (c23) -- (c24) node[labeledhb] {\scriptsize $(8)$};
\path[draw, thick, red] (c24) -- (c25) node[labeledhb] {\scriptsize $(9)$};
\path[draw, thick] (c25) -- (c26) node[labeledhb] 		{\scriptsize $(10)$};

\path[draw, thick] (c21) -- (c31) node[labeledvr] 		{\small $1$};
\path[draw, thick] (c22) -- (c32) node[labeledvr] 		{\small $3$};
\path[draw, thick] (c23) -- (c33) node[labeledvr] 		{\small $2$};
\path[draw, thick] (c24) -- (c34) node[labeledvr] 		{\small $2$};
\path[draw, thick, red] (c25) -- (c35) node[labeledvr] {\small $3$};
\path[draw, thick] (c26) -- (c36) node[labeledvr] 		{\small $1$};
\path[draw, thick] (c21) -- (c31) node[labeledvl] 		{\scriptsize $(11)$};
\path[draw, thick] (c22) -- (c32) node[labeledvl] 		{\scriptsize $(12)$};
\path[draw, thick] (c23) -- (c33) node[labeledvl] 		{\scriptsize $(13)$};
\path[draw, thick] (c24) -- (c34) node[labeledvl] 		{\scriptsize $(14)$};
\path[draw, thick, red] (c25) -- (c35) node[labeledvl] 	{\scriptsize $(15)$};
\path[draw, thick] (c26) -- (c36) node[labeledvl] 		{\scriptsize $(16)$};

\path[draw, thick] (c31) -- (c32) node[labeledha] 		{\small $1$};
\path[draw, thick] (c32) -- (c33) node[labeledha] 		{\small $2$};
\path[draw, thick] (c33) -- (c34) node[labeledha] 		{\small $2$};
\path[draw, thick] (c34) -- (c35) node[labeledha] 		{\small $5$};
\path[draw, thick] (c35) -- (c36) node[labeledha] 		{\small $2$};
\path[draw, thick] (c31) -- (c32) node[labeledhb] 		{\scriptsize $(17)$};
\path[draw, thick] (c32) -- (c33) node[labeledhb] 		{\scriptsize $(18)$};
\path[draw, thick] (c33) -- (c34) node[labeledhb] 		{\scriptsize $(19)$};
\path[draw, thick] (c34) -- (c35) node[labeledhb] 		{\scriptsize $(20)$};
\path[draw, thick] (c35) -- (c36) node[labeledhb] 		{\scriptsize $(21)$};

\path[draw, thick] (c34) -- (c44) node[labeledvr] 		{\small $1$};
\path[draw, thick] (c35) -- (c45) node[labeledvr] 		{\small $2$};
\path[draw, thick] (c36) -- (c46) node[labeledvr] 		{\small $1$};
\path[draw, thick] (c34) -- (c44) node[labeledvl] 		{\scriptsize $(22)$};
\path[draw, thick] (c35) -- (c45) node[labeledvl] 		{\scriptsize $(23)$};
\path[draw, thick] (c36) -- (c46) node[labeledvl] 		{\scriptsize $(24)$};

\path[draw, thick] (c44) -- (c45) node[labeledha] 		{\small $1$};
\path[draw, thick] (c45) -- (c46) node[labeledha] 		{\small $1$};
\path[draw, thick] (c44) -- (c45) node[labeledhb] 		{\scriptsize $(25)$};
\path[draw, thick] (c45) -- (c46) node[labeledhb] 		{\scriptsize $(26)$};
\end{tikzpicture}}
  \end{center}
  \caption{A rescaled pseudocodeword of a toric code, where component $\tilde{\om}_i$ and index $(i)$ are shown next to the $i^{\Rm{th}}$ edge.}
  \label{fig:PSPA_toric9}

\end{figure}


\section{Pseudocodeword-based Decoding}
\label{sec:pcw:decoding:0}

If we want to improve the performance of SPA of quantum stabilizer cycle
codes, or, more generally, quantum stabilizer codes, the first task is to
address the problem mentioned in the proof of Theorem~\ref{thm:SPA_half_d_N}
by breaking the symmetry of the SPA to avoid ending up with pseudocodewords
like the ones in Fig.~\ref{fig:PSPA_toric2}.  


\subsection{Reweighted SPA Decoding}
\label{sec:pcw:decoding:1}

Our first approach is to use the reweighted SPA decoding proposed
in~\cite{wymeersch2011uniformly}, which reweights message
calculations. However, instead of uniformly reweighing the messages, we
randomly select weights from a certain interval. We call the resulting
algorithm randomly reweighted SPA (RR-SPA). Empirically, this method can
improve the performance of SPA decoding of the toric codes, but there is not
much improvement for MacKay's bicycle codes. 


\subsection{Initial-message-reweighted SPA of Quantum Stabilizer Codes}

In order to introduce our second approach, we recall that the SPA is 
based on the log-likelihood ratios (LLRs) $\gamma_i \defeq \log \left(\frac{\Rm{Pr}(E_i=0)}{\Rm{Pr}(E_i=1)}\right)$ and the syndrome $\vs$.
Our second approach is called initial-message-reweighted SPA (IMR-SPA) and described in Algorithm~\ref{alg:rSPA}. The IMR-SPA also runs the SPA, however, with the reweighted LLRs, i.e., $\gamma_i$ is replaced by $\alpha_i \gamma_i$, where $\alpha_i$ is a weighting factor randomly generated from some interval. Empirically, it is 
observed that the RR-SPA and the IMR-SPA have similar performance for the
toric codes. From an analysis point of view, the IMR-SPA may be preferable
compared to the RR-SPA and other approaches like random perturbation~\cite{poulin2008iterative}, because after suitable adapations, we can apply the
Bethe free energy framework~\cite{Yedidia:Freeman:Weiss:05:1, vontobel2013counting, pfister2013relevance} to analyze the IMR-SPA.

We briefly explain why the IMR-SPA helps to improve the performance of SPA
decoding of quantum stabilizer codes. 
Namely, assume that we know, for analysis purposes, the actual
error vector $\tilde{\ve}$. For SPA decoding, using the
LLR vector $\boldV{\gamma}$ with the syndrome $\vs$ is equivalent
to using the LLR vector $\tilde{\boldV{\gamma}}$, where $\tilde{\gamma}_i \defeq
(-1)^{\tilde{e}_i} \gamma_i $, with the syndrome $\mb{0}$.  
SPA decoding succeeds when it converges to the all-zero vector based on the LLR vector
$\tilde{\boldV{\gamma}}$ and the syndrome $\mb{0}$.  The IMR-SPA changes the
LLR vector for the standard SPA from $\tilde{\gamma}_i$ to be $\alpha_i \tilde{\gamma}_i$ and hence may
move some $\tilde{\boldV{\gamma}}$ from the ``bad'' region to the ``good''
region in which the SPA converges to the all-zero vector.


\begin{algorithm}[t] 
\caption{Initial-message-reweighted SPA (IMR-SPA)} 
\label{alg:rSPA}
\begin{algorithmic}[1]
\REQUIRE the syndrome $\vs$, the maximum number of SPA iterations, and the reweighting range $[a,b]$.
\ENSURE $\vv+\cB$. 
\STATE Use SPA to obtain an output vector $\vv$.
\IF {$H \vv^{\tr} = \vs^{\tr}$ (equivalently $\vv \in \vt(\vs)+\cN$)}
	\STATE Output $\vv+\cB$. 
\ELSE
    	\WHILE{$H \vv^{\tr} \neq \vs^{\tr}$}
    		\STATE For the $i^{\Rm{th}}$ variable, randomly generate a weighting factor $\alpha_i \in [a,b]$ and reweight the LLR to the SPA from $\gamma_i$ to be $\alpha_i \gamma_i$. 
    		\STATE Use SPA to obtain an output vector $\vv$.\\ 
    		\COMMENT{Set the max.\ number of trial times if necessary.}
    	\ENDWHILE
    	\STATE Output $\vv+\cB$. 
\ENDIF
\end{algorithmic}
\end{algorithm}



\subsection{Pseudocodeword-based Decoder of Quantum Cycle Codes}

For quantum cycle codes, the IMR-SPA decoding can improve the minimum weight
of errors leading to decoding failures beyond $d_{\cN}/2$, but it is still limited by the problems
mentioned in Theorem~\ref{thm:SPA_d_N}. Therefore, we propose a
pseudocodeword-based decoder abbreviated as SPA+PCWD, which is described in Algorithm~\ref{alg:pSPA}, to further improve the performance of SPA decoding
for quantum cycle codes.
  When SPA decoding fails to output a vector with valid syndrome, 
  we hope to make use of the SPA pseudocodeword to obtain one with valid syndrome.
  There are two difficulties in this task:
  1) the components contributed by codewords from graph covers without a valid syndrome need to be removed;
  2) the components contributed by codewords from graph covers with a valid syndrome are mixed together and need to be separated.

The main idea of the decoder is to first decompose the pseudocodeword $\boldV{\om}$ into a set of paths
and then output a vector $\vv$ with a valid syndrome, where the support of $\vv$ is determined by a collection of paths.
The paths are obtained by starting from an unsatisfied check $s_i$ and 
by always following the edge with the largest possible component of $\boldV{\om}$ for the next step without repetition until reaching an unsatisfied check $s_{i'}$, where the weight of the path is defined as the minimum component of $\boldV{\om}$ on that path.
The contribution of that path from the pseudocodeword is then subtracted and the path is included in the set of candidate paths.
We use a simple example to explain the procedure of Algorithm~\ref{alg:pSPA}.


\begin{example}
Consider a $\llbracket 2L^2,2,L\rrbracket$ toric code of $L = 9$ and $p=0.0123$. 
We obtain an SPA pseudocodeword $\boldV{\om}$ after 100 iterations. 
The rescaled pseudocodeword $\tilde{\boldV{\om}} \defeq \boldV{\om}/0.0836$ is shown in Fig.~\ref{fig:PSPA_toric9}, where edges with component $\tilde{\om}_i<0.005$ are omitted.
First, by Algorithm~\ref{alg:PCWD} we can obtain a set of paths $P$, 
e.g., $P_1 = \{7,8,9,15\}$ and $P_2=\{7,13,19,20\}$ with $\hat{\om}_i = 2$ and $S_i = \{1,2\}$ for $i=1,2$.
Then, Algorithm~\ref{alg:pSPA} picks an arbitrary path $P_{i^{*}}$ since their costs are the same and outputs $\vv+\cB$, where the support of $\vv$ is determined by $P_{i^{*}}$.
\end{example}


\begin{algorithm}[t] 
\caption{Pseudocodeword-based decoder (SPA+PCWD) for quantum cycle codes} 
\label{alg:pSPA} 
\begin{algorithmic}[1]
\REQUIRE the syndrome $\vs$ and the max.\ number of SPA iterations.
\ENSURE $\vv+\cB$. 

\STATE Use SPA to find an output vector $\vv$ and obtain an SPA pseudocodeword $\boldV{\om}$.

\IF {$H \vv^{\tr} = \vs^{\tr}$}
	\STATE Return $\vv+\cB$. 
\ELSE
	\STATE Obtain $P= \{P_i\}$, $\{S_i\}$, and $\hat{\boldV{\om}}$ by Algorithm~\ref{alg:PCWD}.
    	\STATE $\vv \leftarrow \mb{0}$. 
    	\STATE $\cJ \leftarrow \{j\ |\ s_j\!\neq\! 0\}$ \{The index set of unsatisfied checks.\}
    	\WHILE{$\cJ \neq \emptyset$ and $P \neq \emptyset$} 
    		\STATE $i \leftarrow \argmin_{j: P_j\in P} (1\!-\!\hat{\om}_j) \cdot |P_j|$. \{Minimize cost.\}
    		\STATE $v_j \leftarrow 1$ $\forall j \in P_{i}$. \{Update $\vv$ w.r.t. $P_i$.\}
    		\STATE $\cJ \leftarrow \cJ \backslash S_i$. \{Update unsatisfied checks.\}
    		\STATE $P \leftarrow \big\{P_j \in P\ \big|\ S_j\subseteq\cJ,\ P_j \subseteq \{\ell\ |\ v_\ell\!=\!0\}\big\}$. 
    		\COMMENT {Update the set of available paths.}
    	\ENDWHILE \\
\COMMENT{A modification of this algorithm with the above while loop replaced by an LP with cost for each path $P_i$ as $\lambda_i \defeq (1-\hat{\om}_i) \cdot |P_i|$ is referred as SPA+LPPCWD.}
\ENDIF
\end{algorithmic}
\end{algorithm}

\begin{algorithm}[t] 
\caption{Pseudocodeword decomposition (PCWD) for quantum cycle codes} 
\label{alg:PCWD} 
\begin{algorithmic}[1]
\REQUIRE a pseudocodeword $\boldV{\om}$ and the syndrome $\vs$.
\ENSURE A set of paths $P = \{P_i\}$, a set of corresponding end checks $S=\{S_i\}$, and a weight vector $\hat{\boldV{\om}}$.
\STATE $\cJ \leftarrow  \{j\ |\ s_j\!\neq\! 0\}$, $P \leftarrow \emptyset$, and $S \leftarrow \emptyset$.\\

\WHILE{$\cJ \neq \emptyset$}
	\STATE Start from each $s_j$, $j \in \cJ$, and follow the edge with the largest possible component of $\boldV{\om}$ at each step without repetition until reaching $s_{j'}$, $j' \in \cJ$, to obtain a path $P_j$ with weight $\bar{\om}_j \leftarrow \min_{\ell \in P_j} \om_\ell$ and $S_j \leftarrow \{j,j'\}$.
	\STATE $i \leftarrow \argmin_{j: P_j\in P} (1-\bar{\om}_j) \cdot |P_j|$. \{Find min. cost one.\}
	\STATE $\om_\ell \leftarrow \om_\ell - \bar{\om}_{i}$ $\forall \ell \in P_{i}$. \{Subtract $P_i$'s contribution.\}
	\STATE $P \leftarrow P \cup \{ P_{i}\}$, $S \leftarrow S \cup \{ S_{i}\}$, and $\hat{\om}_i \leftarrow \bar{\om}_i$ if $i \neq i'$.
	\COMMENT{Include $P_i$ in $P$ if $P_i$ is a path.}
	\STATE $\cJ \leftarrow \cJ \backslash \{j \in \cJ\ |\ \bar{\om}_j=0\}$. \{Remove isolated checks.\}
\ENDWHILE
\end{algorithmic}
\end{algorithm}


\section{Simulation Results}
\label{sec:simulation:results:1}

Fig.~\ref{fig:pSPA} shows some simulation results of SPA+LPPCWD decoding of
toric codes described in Algorithm~\ref{alg:pSPA}, where we use at most 100
iterations of SPA to obtain SPA pseudocodewords.  According to the simulation
results in~\cite{liu2018neural}, the performance of the original SPA gets
worse as the code block length of toric codes increases.  As shown in
Fig.~\ref{fig:pSPA}, the performance of the SPA+LPPCWD improves as the code
block length of toric codes increases and the SPA+LPPCWD has similar performance
as the neural belief-propagation decoder in~\cite{liu2018neural} and the Markov chain 
Monte Carlo algorithm in~\cite{hutter2014efficient}.  
Fig.~\ref{fig:pSPA_Wt} shows the weight distribution of the decoding errors of 
SPA+PCWD decoding of toric codes, where the minimum weight of errors increases 
as the block length increases. We also
observed that the IMR-SPA and the RR-SPA have similar performance as the
SPA+PCWD for toric codes with $L < 9$, but unfortunately they are limited by
some weight-$4$ errors for $L \geq 9$.

Fig.~\ref{fig:rSPA} shows some simulation results of the IMR-SPA of a
$\llbracket 256,32 \rrbracket$ MacKay's bicycle code with the total row weight
$16$ for $\cN$.  The maximum number of iterations of SPA is 100 and the maximum 
number of IMR trials is 10.  The IMR-SPA achieves lower WER at around $p=10^{-2}$
compared with the the neural belief-propagation decoder~\cite{liu2018neural} for MacKay's bicycle codes with the same parameters.  


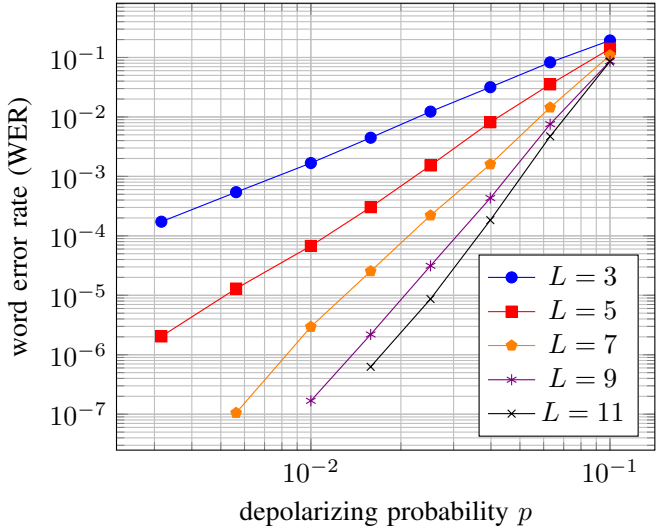
\begin{figure}[t]
  \begin{center}
    \resizebox{1\columnwidth}{!}{\begin{tikzpicture}
\begin{loglogaxis}[
    xlabel={depolarizing probability $p$},
    ylabel={word error rate (WER)},
    legend pos=south east,
    grid=both
]

\addplot[color=blue,mark=*] coordinates { 
	(0.100000000000000,     0.193236714975845) 		
	(0.0630957344480193,   0.0827814569536424)	 	
	(0.0398107170553497,   0.0316505776230416) 	    	
	(0.0251188643150958,   0.0122865216857108)	 	
	(0.0158489319246111,   0.00446937361728754) 		
	(0.0100000000000000,   0.001676628635140460) 	
	(0.0056250000000000,   0.000542264060907099) 	
	(0.0031640625000000,   0.000172855319233525) 	
    };

\addplot[color=red,mark=square*] coordinates { 
	(0.100000000000000,    0.140449438202247) 	    	
	(0.0630957344480193,  0.0353107344632768)	 	
	(0.0398107170553497,  0.00813504169208867) 	     
	(0.0251188643150958,  0.00152843266871986)	     
	(0.0158489319246111,  0.000303367224508431) 	
	(0.0100000000000000,  6.75632079146244e-05) 	
	(0.0056250000000000,  1.29675439880164e-05) 	
	(0.0031640625000000,  2.03315727923810e-06) 	
    };

\addplot[color=orange,mark=pentagon*] coordinates { 
	(0.100000000000000,    0.108932461873638) 	     
	(0.0630957344480193,  0.0144237703735757)	 	
	(0.0398107170553497,  0.00159740581295975)  		
	(0.0251188643150958,  0.000220425995278475)	    
	(0.0158489319246111,  2.53927947944771e-05) 	
	(0.0100000000000000,  2.94948586464412e-06) 	
	(0.0056250000000000,  1.05228556160910e-07) 	
    };
   
\addplot[color=violet,mark=asterisk] coordinates { 
	(0.100000000000000,   0.0865051903114187) 	    	
	(0.0630957344480193, 0.00756572725553244)	 	
	(0.0398107170553497,  0.000432371666684682) 	
	(0.0251188643150958,  3.13327046512335e-05)	 	
	(0.0158489319246111,  2.19520142688093e-06) 	
	(0.0100000000000000,  1.67814196122054e-07) 	
    };
    
\addplot[color=black,mark=x] coordinates { 
	(0.100000000000000,   0.0856164383561644) 	    	
	(0.0630957344480193,  0.00472087808332350)	 	
	(0.0398107170553497,  0.000183683519629593) 	
	(0.0251188643150958,  8.66241337586624e-06)	 	
	(0.0158489319246111,  6.26973791868526e-07) 	
    };
    
\legend{$L = 3$, $L = 5$, $L = 7$, $L = 9$, $L = 11$}
 
\end{loglogaxis}
\end{tikzpicture}}
  \end{center}
  \caption{Simulation results of SPA+LPPCWD for toric codes.}
  \label{fig:pSPA}
\end{figure}


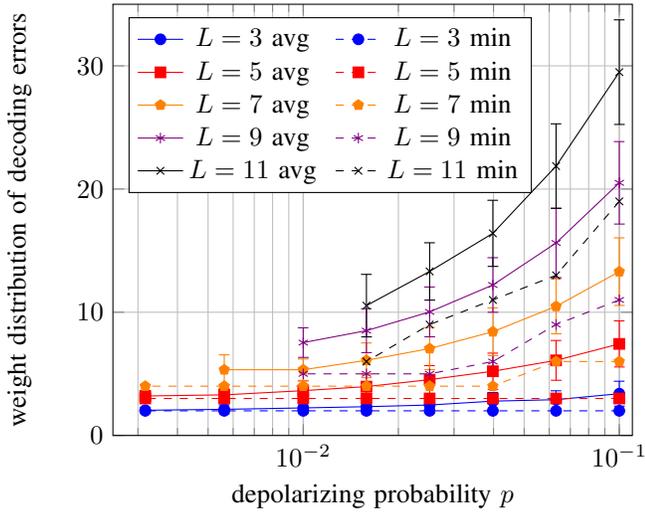
\begin{figure}[t]
  \begin{center}
    \resizebox{1\columnwidth}{!}{\begin{tikzpicture}
\begin{semilogxaxis}[
    xlabel={depolarizing probability $p$},
    mark options={solid},
    xmin=0.0025, xmax=0.11,
    ymin=0, ymax=35,
    ylabel={weight distribution of decoding errors},
    grid=both,
    legend columns=2,
    legend style={
    legend pos=north west,
    /tikz/column 2/.style={
                column sep=5pt,},
	}
]

\addplot[color=blue,mark=*,error bars/.cd, y dir=both,y explicit] coordinates { 
	(0.100000000000000,     3.39000000000000) +- (1.01144207712399,1.01144207712399)		
	(0.0630957344480193,   2.90500000000000)	+- (0.720116545564660,0.720116545564660) 	
	(0.0398107170553497,   2.78500000000000) 	+- (0.700878845292748,0.700878845292748)    	
	(0.0251188643150958,   2.47000000000000)	+- (0.641214300288502,0.641214300288502) 	
	(0.0158489319246111,   2.33500000000000) 	+- (0.586939083242193,0.586939083242193)	
	(0.0100000000000000,   2.23500000000000) 	+- (0.459161350429937,0.459161350429937)	
	(0.0056250000000000,   2.12000000000000) 	+- (0.355290241603811,0.355290241603811)	
	(0.0031640625000000,   2.04500000000000) 	+- (0.230740560002247,0.230740560002247)	
    };
\addplot[dashed,color=blue,mark=*] coordinates { 
	(0.100000000000000,     2)	
	(0.0630957344480193,   2) 	
	(0.0398107170553497,   2) 	
	(0.0251188643150958,   2) 	
	(0.0158489319246111,   2)	
	(0.0100000000000000,   2)	
	(0.0056250000000000,   2)	
	(0.0031640625000000,   2)	
    };
    
\addplot[color=red,mark=square*,error bars/.cd, y dir=both,y explicit] coordinates { 
	(0.100000000000000,    7.43000000000000) +- (1.87152693417077,1.87152693417077)	    	
	(0.0630957344480193,  6.09500000000000) +- (1.61524614238899,1.61524614238899)	 	
	(0.0398107170553497,  5.21000000000000) +- (1.47198290354939,1.47198290354939)	     
	(0.0251188643150958,  4.53000000000000) +- (1.13823454939974,1.13823454939974)        
	(0.0158489319246111,  3.96500000000000) +- (1.02421682299331,1.02421682299331)		
	(0.0100000000000000,  3.63500000000000) +- (0.790505849116525,0.790505849116525)		
	(0.0056250000000000,  3.29500000000000) +- (0.537979656262914,0.537979656262914)		
	(0.0031640625000000,  3.20300751879699) +- (0.439686501544008,0.439686501544008)		
    };
\addplot[dashed,color=red,mark=square*] coordinates { 
	(0.100000000000000,     3)	
	(0.0630957344480193,   3) 	
	(0.0398107170553497,   3) 	
	(0.0251188643150958,   3) 	
	(0.0158489319246111,   3)	
	(0.0100000000000000,   3)	
	(0.0056250000000000,   3)	
	(0.0031640625000000,   3)	
    };
    
\addplot[color=orange,mark=pentagon*,error bars/.cd, y dir=both,y explicit] coordinates { 
	(0.100000000000000,    13.2950000000000) +- (2.74328325461943,2.74328325461943)	     
	(0.0630957344480193,  10.4800000000000) +- (2.22809831707507,2.22809831707507) 		
	(0.0398107170553497,  8.42000000000000) +- (1.93710139358920,1.93710139358920) 		
	(0.0251188643150958,  7.05000000000000) +- (1.71235533812936,1.71235533812936)    		
	(0.0158489319246111,  6.09722222222222) +- (1.41084433602072,1.41084433602072)		
	(0.0100000000000000,  5.33333333333333) +- (0.868114732282431,0.868114732282431)		
	(0.0056250000000000,  5.33333333333333) +- (1.21106014163900,1.21106014163900)		
    };
\addplot[dashed,color=orange,mark=pentagon*] coordinates { 
	(0.100000000000000,     6)	
	(0.0630957344480193,   6) 	
	(0.0398107170553497,   4) 	
	(0.0251188643150958,   4) 	
	(0.0158489319246111,   4)	
	(0.0100000000000000,   4)	
	(0.0056250000000000,   4)	
	(0.0031640625000000,   4)	
    };

\addplot[color=violet,mark=asterisk,error bars/.cd, y dir=both,y explicit] coordinates { 
	(0.100000000000000,   20.5050000000000)  +- (3.34753715813964, 3.34753715813964)	    	
	(0.0630957344480193, 15.6250000000000)  +- (2.81852718230407, 2.81852718230407)	 	
	(0.0398107170553497,  12.2150000000000) +- (2.21433770172628, 2.21433770172628)		
	(0.0251188643150958,  10.0291262135922) +- (2.02172664838170, 2.02172664838170)	 	
	(0.0158489319246111,  8.50617283950617) +- (1.78271882801329, 1.78271882801329)		
	(0.0100000000000000,  7.53571428571429) +- (1.20130000129930, 1.20130000129930)		
    };
\addplot[dashed,color=violet,mark=asterisk] coordinates { 
	(0.100000000000000,     11)	
	(0.0630957344480193,   9) 	     
	(0.0398107170553497,   6) 		
	(0.0251188643150958,   5) 		
	(0.0158489319246111,   5)		
	(0.0100000000000000,   5)		
    };    
       
\addplot[color=black,mark=x,error bars/.cd, y dir=both,y explicit] coordinates { 
	(0.100000000000000,   29.4900000000000)  +- (4.24499705535822, 4.24499705535822)	    	
	(0.0630957344480193,  21.8650000000000) +- (3.43010247880080, 3.43010247880080) 		
	(0.0398107170553497,  16.4051724137931) +- (2.68321172988758, 2.68321172988758)	    
	(0.0251188643150958,  13.3132530120482) +- (2.32655775689748, 2.32655775689748) 		
	(0.0158489319246111,  10.5384615384615) +- (2.53690706094990, 2.53690706094990)		
    };
    
\addplot[dashed,color=black,mark=x] coordinates { 
	(0.100000000000000,     19)	
	(0.0630957344480193,   13) 	
	(0.0398107170553497,   11) 	
	(0.0251188643150958,   9) 		
	(0.0158489319246111,   6)		
    };
\legend{$L = 3$ avg, $L = 3$ min, $L = 5$ avg, $L = 5$ min, $L = 7$ avg, $L = 7$ min, $L = 9$ avg, $L = 9$ min, $L = 11$ avg,  $L = 11$ min
}
 
\end{semilogxaxis}
\end{tikzpicture}}
  \end{center}
  \caption{Weight distribution of decoding errors of SPA+PCWD for toric codes (solid: average, dashed: minimum).}
  \label{fig:pSPA_Wt}
\end{figure}


\begin{figure}[t]
  \begin{center}
    \resizebox{1\columnwidth}{!}{\begin{tikzpicture}
\begin{loglogaxis}[
    xlabel={depolarizing probability $p$},
    ylabel={word error rate (WER)},
    legend pos=south east,
    grid=both
]

\addplot[color=blue,mark=*] coordinates { 
	(0.100000000000000,     1) 		
	(0.0630957344480193,   0.876971608832808)	 	
	(0.0398107170553497,   0.370260955194486) 	    	
	(0.0251188643150958,   0.0882872125019470)	 	
	(0.0158489319246111,   0.0181824207874851) 		
	(0.0100000000000000,   0.00378668941979522) 		
	(0.0056250000000000,   0.000500835840855901) 	
    };

\addplot[color=red,mark=square*] coordinates { 
	(0.100000000000000,    0.985221674876847) 	    	
	(0.0630957344480193,  0.630914826498423)	 		
	(0.0398107170553497,  0.0984736582964057) 	     
	(0.0251188643150958,  0.00519183842998806)	     
	(0.0158489319246111,  0.000213442927217288) 	
	(0.0100000000000000,  1.87713310580205e-05) 	
	(0.0056250000000000,  2.67469073888332e-06) 	
    };
    
\legend{SPA, IMR-SPA}
 
\end{loglogaxis}
\end{tikzpicture}}
  \end{center}
  \caption{Simulation results for SPA and IMR-SPA decoding 
  of a $\llbracket 256,32 \rrbracket$ MacKay's bicycle code.}
  \label{fig:rSPA}
\end{figure}
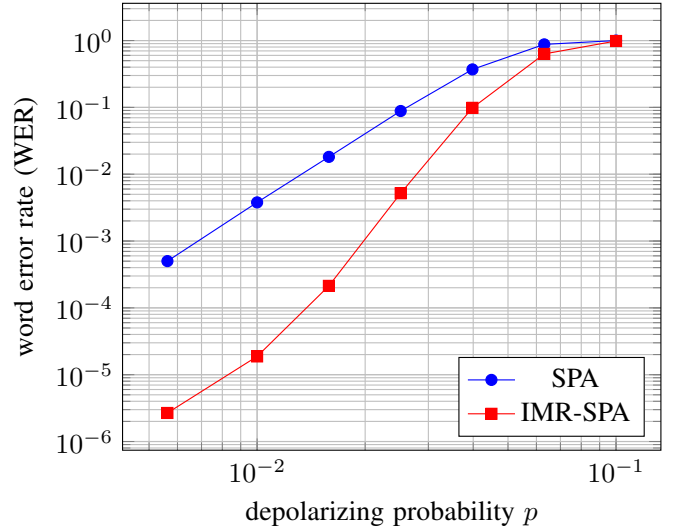


\appendices
\section{Proof of Theorem~\ref{thm:ND_decoder}}
\label{app:ND_decoder_proof}

Proof sketch: 
The idea is
  to show each coset of $\cN$ contains at most one vector of weight less than
  or equal to $t_{\cN}$ and there exists some coset of $\cN$ containing two
  vectors of weights less than or equal to $t_{\cN}+1$, one of which is a
  decoding error. 
  
Let the set of decoding errors and the minimum weight of decoding errors for $\cD_{\Rm{A}}^{\Rm{B}}$ be 
$\cE_{\Rm{A}}^{\Rm{B}} \defeq \GF{2}^{2n}\backslash \left(\cup_{\vs}\cD_{\Rm{A}}^{\Rm{B}}(\vs)\right)$ and 
$d_{\Rm{A}}^{\Rm{B}} \defeq \min_{\vv \in \cE_{\Rm{A}}^{\Rm{B}}} \Wt(\vv)$.
We show that $d_{\Rm{ND}}^{\Rm{ML}} \geq t_{\cN}+1$ and $d_{\Rm{ND}}^{\Rm{ML}}
\leq t_{\cN}+1$.

\begin{itemize}

\item ($d_{\Rm{ND}}^{\Rm{ML}} \geq t_{\cN}+1$) For any syndrome
  $\vs\in\GF{2}^{n-k}$, there is at most one $\vv\in\vt(\vs)+\cN$ such that
  $\Wt(\vv) \leq t_{\cN}$, otherwise suppose there are $\vv_1,\
  \vv_2\in\vt(\vs)+\cN$ such that $\Wt(\vv_1),\ \Wt(\vv_2) \leq t_{\cN}$ and
  then $\vv_1+\vv_2\in\cN$ with $\Wt(\vv_1+\vv_2)\leq 2t_{\cN}<d_{\cN}$. A
  contradiction arises.

  Hence for all the vectors with weight no more than $t_{\cN}$, they must have
  distinct syndromes and they are all in
  $\cup_{\vs}\cD_{\Rm{ND}}^{\Rm{ML}}(\vs)$ and not in
  $\cE_{\Rm{ND}}^{\Rm{ML}}$, which implies $d_{\Rm{ND}}^{\Rm{ML}} \geq
  t_{\cN}+1$.

\item ($d_{\Rm{ND}}^{\Rm{ML}} \leq t_{\cN}+1$) Since $d_{\cN} =
  \min_{\vv\in\cN} \Wt(\vv)$, there exists $\vv\in\cN$ such that
  $\Wt(\vv)=d_{\cN}$.  There exist $\vv_1,\ \vv_2\in\GF{2}^{2n}$ such that
  $\vv_1 + \vv_2 =\vv$, $\Wt(\vv_2)=t_{\cN}+1$, and $\Wt(\vv_1) = d_{\cN} -
  \Wt(\vv_2) \leq d_{\cN} - (t_{\cN}+1) \leq t_{\cN}+1$.  Then $\vv_1,\
  \vv_2\in \vt(\vs)+\cN$ for some $\vs$ and at most one of them can be in
  $\cD_{\Rm{ND}}^{\Rm{ML}}(\vs)$.  Hence $\vv_i\in\cE_{\Rm{ND}}^{\Rm{ML}}$ for
  some $i$ and $d_{\Rm{ND}}^{\Rm{ML}} \leq \Wt(\vv_i) \leq t_{\cN}+1$.
\end{itemize}


\section{Proof of Theorem~\ref{thm:DD_decoder}}
\label{app:DD_decoder_proof}
Let the set of decoding errors and the minimum weight of decoding errors for $\cD_{\Rm{A}}^{\Rm{B}}$ be 
$\cE_{\Rm{A}}^{\Rm{B}} \defeq \GF{2}^{2n}\backslash \left(\cup_{\vs}\cD_{\Rm{A}}^{\Rm{B}}(\vs)\right)$ and 
$d_{\Rm{A}}^{\Rm{B}} \defeq \min_{\vv \in \cE_{\Rm{A}}^{\Rm{B}}} \Wt(\vv)$.
We first show that $d_{\Rm{D}}^{\Rm{ML}*},d_{\Rm{D}}^{\Rm{ML}} \leq t+1$ and then
$d_{\Rm{D}}^{\Rm{ML}*} \geq t+1$.
\begin{itemize}

\item ($d_{\Rm{D}}^{\Rm{ML}*},d_{\Rm{D}}^{\Rm{ML}} \leq t+1$) Since $d \defeq
  \min_{\vv\in\cN\backslash\cB} \Wt(\vv)$, there exists
  $\vv\in\cN\backslash\cB$ such that $\Wt(\vv)=d$.  There exist $\vv_1,\
  \vv_2\in\GF{2}^{2n}$ such that $\vv_1 + \vv_2 =\vv$, $\Wt(\vv_2)=t+1$, and
  $\Wt(\vv_1) = d - \Wt(\vv_2) \leq d_{\cN} - (t+1) \leq t+1$.  Then we have
  $\vv_1\in\vt(\vs)+\boldV{\ell}_1+\cB$ and
  $\vv_2\in\vt(\vs)+\boldV{\ell}_2+\cB$, for some $\vs$ and
  $\boldV{\ell}_1\neq\boldV{\ell}_2 \in\cN$.  Then at most one of them can be
  in $\cD_{\Rm{D}}^{\Rm{ML}*}(\vs)$ or $\cD_{\Rm{D}}^{\Rm{ML}}(\vs)$.  
  Hence $\vv_{i_1}\in\cE_{\Rm{D}}^{\Rm{ML}*}$ and $\vv_{i_2}\in\cE_{\Rm{D}}^{\Rm{ML}}$
  for some $i_1$, $i_2$ and $d_{\Rm{D}}^{\Rm{ML}*}\leq \Wt(\vv_{i_1}) \leq t_{\cN}+1$ and $d_{\Rm{D}}^{\Rm{ML}} \leq \Wt(\vv_{i_2}) \leq t_{\cN}+1$.

\item ($d_{\Rm{D}}^{\Rm{ML}*} \geq t+1$) For any syndrome
  $\vs\in\GF{2}^{n-k}$, there is at most one coset of $\cB$ in $\vt(\vs)+\cN$
  containing vectors of weights smaller or equal to $t$, otherwise suppose
  there are $\vv_1,\ \vv_2\in \vt(\vs)+\cN$ with $\Wt(\vv_1),\ \Wt(\vv_2) \leq
  t$ such that $\vv_1\in\vt(\vs)+\boldV{\ell}_1+\cB$ and
  $\vv_2\in\vt(\vs)+\boldV{\ell}_2+\cB$ for
  $\boldV{\ell}_1\neq\boldV{\ell}_2\in\cN$, and then
  $\vv_1+\vv_2\in\cN\backslash \cB$ with $\Wt(\vv_1+\vv_2)\leq 2t<d$. A
  contradiction arises.  Hence all vectors of weight no larger than $t$ are in
  $\cup_{\vs}\cD_{\Rm{D}}^{\Rm{ML}*}(\vs)$ and not in
  $\cE_{\Rm{D}}^{\Rm{ML}}$, which implies $d_{\Rm{D}}^{\Rm{ML}*} \geq t+1$.
\end{itemize}
The main contribution of the probabilities of cosets comes from the vectors with minimum weights as $p \rightarrow 0$.
Hence $ d_{\Rm{D}}^{\Rm{ML}} \geq d_{\Rm{D}}^{\Rm{ML}*} \geq t+1$ which implies $d_{\Rm{D}}^{\Rm{ML}}= t+1$.

\bibliographystyle{IEEEtran}
\bibliography{ref}

\end{document}